\documentclass[11pt,letterpaper,english]{article}

\usepackage[utf8]{inputenc}
\usepackage{amsmath}
\usepackage{amssymb}
\usepackage{amsthm}
\usepackage{amsfonts}
\usepackage{color}
\usepackage{graphicx}
\usepackage{url}
\usepackage{xargs}
\usepackage[numbers,comma]{natbib}
\usepackage{paralist}
\usepackage{float}
\usepackage{lineno}
\usepackage{pgf}

\usepackage[labelformat=simple]{subcaption} 
\usepackage[font=small,width=0.95\textwidth]{caption}
\usepackage[pdftex]{hyperref}
\usepackage{kvoptions-patch}
\usepackage{cleveref}

\graphicspath{{fig/}{./fig/}}

\AtBeginDocument{%
  \maketitle
  \hypersetup{%
    pdftitle = {On the $\mathcal{O}_\beta$-hull of a planar point set},
    pdfauthor = {Carlos Alegr\'{i}a-Galicia, David Orden, Carlos Seara, Jorge Urrutia}
  }
}

\newtheorem{theorem}{Theorem}
\newtheorem{lemma}{Lemma}
\newtheorem{obs}[theorem]{Observation}

\theoremstyle{definition}
\newtheorem{problem}{Problem}

\newcommandx*{\bset}[1][usedefault, 1=\beta]{\mathcal{O}_{#1}}
\newcommandx*{\bhull}[1][usedefault, 1=\beta]{\mathcal{O}_{#1}\text{-hull}}
\newcommandx*{\bhullp}[2][usedefault, 1=P, 2=\beta]{\mathcal{O}_{#2}\mathcal{H}({#1})}

\newcommandx*{\cset}[2][usedefault, 1=k, 2=\beta]{\left( {#1},{#2} \right)}
\newcommandx*{\csetc}[3][usedefault, 1=k, 2=\beta, 3=\theta]{\mathcal{C}_{{#1},{#2}}({#3})}

\newcommand{\perim}{\operatorname{perim}}
\newcommand{\area}{\operatorname{area}}
\newcommandx*{\polygon}[1][usedefault, 1=\beta]{\mathcal{P}(#1)}
\newcommandx*{\triangles}[2][usedefault, 1=\beta,2=i]{\triangle_{#2}(#1)}
\newcommand*{\parallelogram}[1][]{%
  \pgfpicture\pgfsetroundjoin
    \pgftransformxslant{.6}%
    \pgfpathrectangle{\pgfpointorigin}{\pgfpoint{.60em}{.65em}}%
    \pgfusepath{stroke,#1}%
  \endpgfpicture}
\newcommandx*{\parallelograms}[2][usedefault, 1=\beta,2=j]{\parallelogram_{#2}(#1)}

\title{On the $\bhull$ of a planar point set\footnote{In memorial of professor Ferran Hurtado, inspirational friend and colleague, acknowledging his key contribution to the development of Computational Geometry.}}

\author{
  Carlos Alegr\'{i}a-Galicia
  \thanks{Posgrado en Ciencia e Ingenier\'{i}a de la Computaci\'on,
    Universidad Nacional Aut\'onoma de M\'exico, {\tt
      alegria\_c@uxmcc2.iimas.unam.mx}. Research supported by H2020-MSCA-RISE project 73499 - CONNECT.}
  \and
  David Orden
  \thanks{Departamento de F\'{i}sica y Matem\'aticas, Universidad de
    Alcal\'a, Spain, {\tt david.orden@uah.es}. Research supported by MINECO Projects
    MTM2014-54207 and TIN2014-61627-EXP, TIGRE5-CM Comunidad de Madrid Project S2013/ICE-2919, and H2020-MSCA-RISE project 73499 - CONNECT.}
  \and
  Carlos Seara
  \thanks{Departament de Matem\`atiques, Universitat
    Polit\`ecnica de Catalunya, Spain, {\tt
      carlos.seara@upc.edu}. Research supported by projects Gen. Cat. DGR 2014SGR46,
      MINECO MTM2015-63791-R, and H2020-MSCA-RISE project 73499 - CONNECT.}
  \and
  Jorge Urrutia
  \thanks{Instituto de Matem\'aticas, Universidad Nacional Aut\'onoma
    de M\'exico, {\tt urrutia@matem.unam.mx}. Research supported by SEP-CONACYT 80268, PAPPIIT IN102117 Programa de Apoyo a la Investigaci\'on e Innovaci\'on Tecnol\'ogica UNAM, and H2020-MSCA-RISE project 73499 - CONNECT.}}

\date{}

\begin{document}

\begin{abstract}
We study the $\bhull$ of a planar point set, a generalization of the
Orthogonal Convex Hull where the coordinate axes form an angle~$\beta$.
Given a set $P$ of $n$ points in the plane, we show how to maintain
the $\bhull$ of~$P$ while $\beta$ runs from $0$ to~$\pi$ in
$\Theta(n\log n)$ time and $O(n)$ space.
With the same complexity, we also find the values of~$\beta$
that maximize the area and the perimeter of the $\bhull$ and,
furthermore, we find the value of~$\beta$ achieving the best fitting
of the point set $P$ with a two-joint chain of alternate interior angle~$\beta$.
\end{abstract}

\section{Introduction}\label{sec:intro}

Let $\bset$ be a set of two lines with slopes $0$ and $\tan(\beta)$,
where $0 < \beta < \pi$. A region in the plane is said to be
\emph{$\bset$-convex}, if its intersections with all translations of
any line in $\bset$ are either empty or connected. An
\emph{$\bset$-quadrant} is a translation of one of the
($\bset$-convex) open regions that result from subtracting the lines
in $\bset$ from the plane. We call the quadrants of $\bset$
as \emph{top-right}, \emph{top left}, \emph{bottom-right}, and \emph{bottom-left}
according to their position with respect to the elements of $\bset$,
see Figure~\ref{intro:fig:bhull}(a).
Let $P$ be a set of $n$ points, and
$\mathcal{Q}$ the set of all $\bset$-quadrants that are
\emph{$P$-free}; i.e., that contain no elements of $P$. The
\emph{$\bhull$} of $P$ is the set

$$\displaystyle
\mbox{\ensuremath{\bhullp}}=\mathbb{R}^{2}-\underset{q\in\mathcal{Q}}{\bigcup}q
$$
of points in the plane belonging to all connected
supersets of $P$ which are
$\bset$-convex~\cite{alegria_2014,ottmann_1984}. See
\Cref{intro:fig:bhull}(b).

\begin{figure}[ht]
  \centering
  \subcaptionbox{\label{intro:fig:bhull:1}}
  {\includegraphics[scale=1.1]{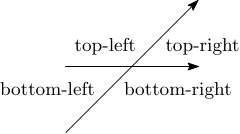}}
  \hspace{2cm}
  \subcaptionbox{\label{intro:fig:bhull:2}}
  {\includegraphics[scale=0.5]{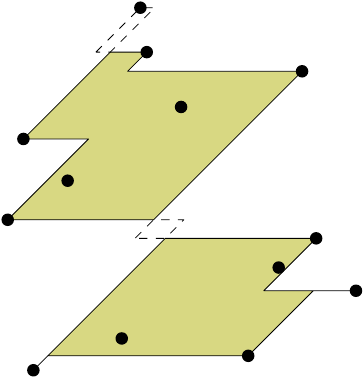}}
  \caption{(a) A set $\bhull$, the top-right, top-left, bottom right, and bottom left quadrants. (b) The corresponding $\bhull$ of a point set.}
  \label{intro:fig:bhull}
\end{figure}

The concept of $\bset$-convexity stemmed from the notion of
\emph{restricted orientations}~\cite{guting_thesis_1983}, where
geometric objects comply with a property (or a set of properties)
related to some fixed set of lines. Researchers have extensively
studied this notion by considering restricted-oriented
polygons~\cite{guting_thesis_1983}, proximity~\cite{widmayer_1987},
visibility~\cite{schuierer_thesis_1991}, and both restrictions and
generalizations of $\bset$-convexity. The particular case of
\emph{orthogonal convexity}~\cite{rawlins_1988} considers $\beta$ to
be fixed at $\frac{\pi}{2}$. In the more general
\emph{$\mathcal{O}$-convexity}~\cite{rawlins_1987,rawlins_1988},
$\bset$ is replaced by a (possibly infinite) set of lines with
arbitrary orientations. Other restricted-oriented notions of convexity
include \emph{$D$-convexity}~\cite{franek_2009} and
\emph{$\mathcal{O}$-convexity}~\cite{rawlins_thesis_1987}. The former
is based in a functional (rather than set-theoretical) definition,
while the latter (unlike $\bset$-convexity) always leads to connected
sets. For a comprehensive compilation of studies on the area please
refer to~\citet{fink_2004}. Some recent computational results can be
found in~\cite{minimum-area_2012,alegria_2013,alegria_2014,pelaez_2013}.

In this paper, we solve the problem of maintaining the combinatorial structure
of $\bhullp$ while $\beta$ goes from $0$ to $\pi$, and apply
this result to some optimization problems. Following the lines
of~\citet{bae_2009}, we find the values of $\beta$ that maximize the
area and the perimeter of $\bhullp$. In addition, we include an appendix extending the results
from~\citet{fitting_2011} to fit a two-joint not-necessarily orthogonal polygonal
chain to a point set. See \Cref{intro:fig:optim}.

\begin{figure}[ht]
  \centering
  \subcaptionbox{\label{bhull:fig:optim:1}}
  {\includegraphics[scale=0.85]{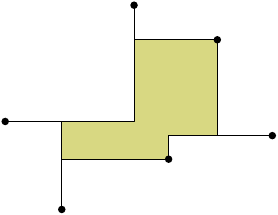}}
 \hspace{-0.2cm}
  \subcaptionbox{\label{bhull:fig:optim:2}}
  {\includegraphics[scale=0.85]{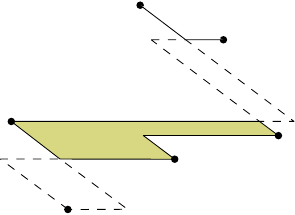}}
  \hspace{-0.4cm}
  \subcaptionbox{\label{bhull:fig:optim:3}}
  {\includegraphics[scale=0.85]{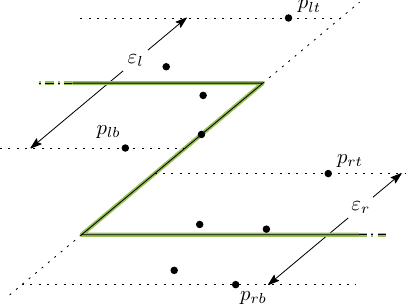}}
  \caption{\subref{bhull:fig:optim:1}
    $\bhullp[P][\frac{\pi}{2}]$. \subref{bhull:fig:optim:2}
    $\bhullp[P][\beta_0]$, where $\beta_0 > \frac{\pi}{2}$.
    \subref{bhull:fig:optim:3} A two-joint non-orthogonal polygonal
    chain fitting a point set.}
  \label{intro:fig:optim}
\end{figure}

In all cases, our general approach is to perform an angular sweep. We
first discretize the set $\{\beta:\beta \in (0,\pi)\}$ into a linear
sequence of increasing angles
$\{\beta_{1},\beta_{2},\ldots,\beta_{O(n)}\}$. While $\beta$ runs from
$0$ to $\pi$, each $\beta_i$ corresponds to an angle where there is a
change in the combinatorial structure of $\bhullp$. We then solve the
particular problem for any $\beta \in [\beta_{1}, \beta_{2})$ in
$O(n\log n)$ time, and show how to update our solution in logarithmic time in the
subsequent intervals $[\beta_{i},\beta_{i+1})$. All our algorithms run
in $O(n\log n)$ time and $O(n)$ space.

\paragraph{Outline of the paper.}

In~\Cref{sec:bhull} we show how to maintain the $\bhull$ of $P$ while
$\beta$ goes from $0$ to $\pi$. In~\Cref{sec:applications} we extend
this result to solve the optimization problems we mentioned above. We
end in~\Cref{sec:conclusions} with our concluding remarks.

\section{The $\bhull$ of $P$}\label{sec:bhull}

In this section we introduce definitions that are central to our
results. We also show how to compute $\bhullp$ for a fixed value of
$\beta$, and how to maintain its combinatorial structure while $\beta$ runs
from $0$ to $\pi$.

\subsection{Preliminaries}\label{sec:bhull:preliminaries}

For the sake of simplicity, we
will assume $P$ to have no three colinear points, and no pair of
points on a horizontal line.
Consider the region $\mathcal{R}$ obtained by removing from the plane
all top-right $\bset$-quadrants free of elements of $P$. The
\emph{top-right $\bset$-staircase} of $P$ is the directed polygonal
chain formed by the segment of the boundary of $\mathcal{R}$ that
starts at the rightmost and ends at the topmost vertex (element of $P$
that lies over the boundary) of $\bhullp$, with respect to the
coordinate system defined by the lines in $\bset$.~We further define
the \emph{top-left}, \emph{bottom-left}, and \emph{bottom-right}
$\bset$-staircases in a similar way. See \Cref{bhull:fig:staircases}.

\begin{figure}[ht]
  \centering
  \subcaptionbox{\label{bhull:fig:staircases:1}}
  {\includegraphics{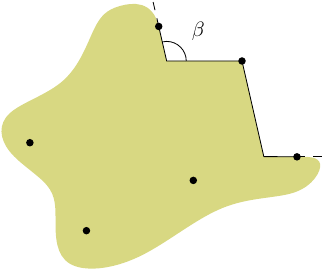}}
  \hspace{1.5cm}
  \subcaptionbox{\label{bhull:fig:staircases:2}}
  {\includegraphics{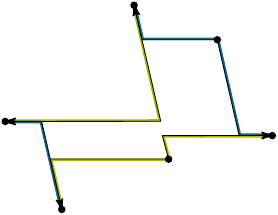}}
  \caption{\subref{bhull:fig:staircases:1} Construction of the
    top-right $\bset$-staircase. \subref{bhull:fig:staircases:2} The
    four $\bset$-staircases of $P$.}
  \label{bhull:fig:staircases}
\end{figure}

\begin{obs}\label{bhull:obs:maximal}
  A point in $P$ is a vertex of $\bhullp$ if, and only if, it is the
  apex of at least one $P$-free $\bset$-quadrant free of elements of
  $P$. Conversely, a point in the plane lies in the interior of
  $\bhullp$ if, and only if, every $\bset$-quadrant with apex on it
  contains at least one point in $P$.
\end{obs}

We say that an $\bset$-quadrant is \emph{maximal} if its boundary
joins two consecutive elements in the sequence of vertices found
while traversing an $\bset$-staircase in its corresponding
direction. Two $\bset$-quadrants are \emph{opposite} to each other if,
after placing their apices over a common point, their rays bound
opposite angles. Similarly, we say that two $\bset$-staircases are
opposite to each other, if they were constructed using opposite
$\bset$-quadrants. It is easy to see that $\bhullp$ is disconnected
when the intersection of two opposite maximal $\bset$-quadrants is not
empty. In such case we say that both $\bset$-quadrants \emph{overlap},
and refer to their intersection as an \emph{overlapping region}. See
the regions bounded by dashed lines in
\Cref{intro:fig:bhull:2,bhull:fig:optim:2}.

\begin{obs}\label{bhull:obs:staircases}
  Non-opposite $\bset$-staircases cannot generate overlapping
  regions. Moreover, only one pair of $\bset$-staircases can intersect
  at the same time.
\end{obs}

We will specify $\bhullp$ in terms of its vertices and its overlapping
regions. From \Cref{bhull:obs:maximal}, the set of vertices of
$\bhullp$ is the set of maximal elements of $P$ under vector
dominance~\cite{theta-maxima_1999}. Thus they can be computed for a
fixed value of $\beta$ in $\Theta(n \log n)$ time and $O(n)$
space~\cite{kung_1975,preparata_1985}. Note that $\bset$-staircases
are monotone with respect to both lines in $\bset$ (they could not
bound $\bset$-convex regions otherwise), so any pair of them intersect
with each other at most a linear number of times. From
\Cref{bhull:obs:staircases}, in a fixed value of~$\beta$ there is at
most a linear number of overlapping regions. Thus, if the vertices of
$\bhullp$ are sorted according to either the $x$- or the $y$-axis, we
can compute from them the set of overlapping regions in linear
time. We get then the following theorem where the $\Omega(n\log n)$
time lower bound comes from the fact that from $\bhullp$ we can compute
the convex hull of $P$ in linear time.

\begin{theorem}\label{intro:thm:fixed_computation}
  For a fixed value of $\beta$, the sets of vertices and overlapping
  regions of $\bhullp$ can be computed in $\Theta(n \log n)$ time and
  $O(n)$ space.
\end{theorem}

\subsection{The angular sweep}\label{sec:bhull:sweep}

The $\bhull$ of $P$ is shown in \Cref{bhull:fig:initial_config} at the
\emph{initial increasing configuration}, that is, where $\beta$ is
equal to an angle $\beta_I = 0 + \varepsilon$ for a small enough
$\varepsilon$. Note that every point in $P$ is the apex of a $P$-free
$\bset$-quadrant, and is thus contained in at least one
$\bset$-staircase: both top-right and bottom-left $\bset$-staircases
contain the whole set $P$, and the top-left and bottom-right
$\bset$-staircases are formed respectively, by the topmost and
bottom-most points in $P$. Also, the intersection between the
top-right and bottom-left $\bset$-staircases generate a linear number
of overlapping regions.

\begin{figure}[ht]
  \centering
  \begin{minipage}{0.9\textwidth}
    \centering
    \includegraphics[clip=true,trim=3cm 0 3cm 0]{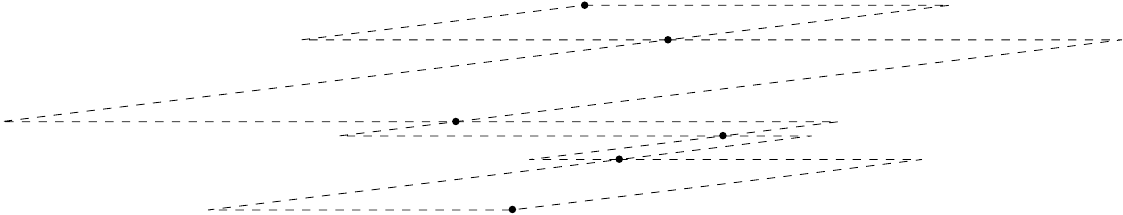}
    \caption{The initial increasing configuration.}
    \label{bhull:fig:initial_config}
  \end{minipage}
\end{figure}

By performing an \emph{increasing sweep} (where $\beta$ goes from $0$
to $\pi$), the initial increasing configuration is gradually
transformed to the \emph{initial decreasing configuration}, where
$\beta$ is equal to a value $\beta_D = \pi - \varepsilon$ for a small
enough $\varepsilon$ (see \Cref{bhull:fig:final_config}). At this
configuration, the top-left and bottom-right $\bset$-staircases
contain $P$ and generate a linear number of overlapping regions, and
the top-right and bottom-left $\bset$-staircases contain respectively,
the topmost and bottom-most points in $P$. Clearly, the converse of
the above discussion holds: from the initial decreasing configuration,
a \emph{decreasing sweep} (where $\beta$ goes from $\pi$ to $0$) will
gradually transform $\bhullp[P][\beta_D]$ into $\bhullp[P][\beta_I]$.

\begin{figure}[ht]
  \centering
  \begin{minipage}{0.9\textwidth}
    \centering
    \includegraphics[clip=true,trim=4cm 0 4cm 0]{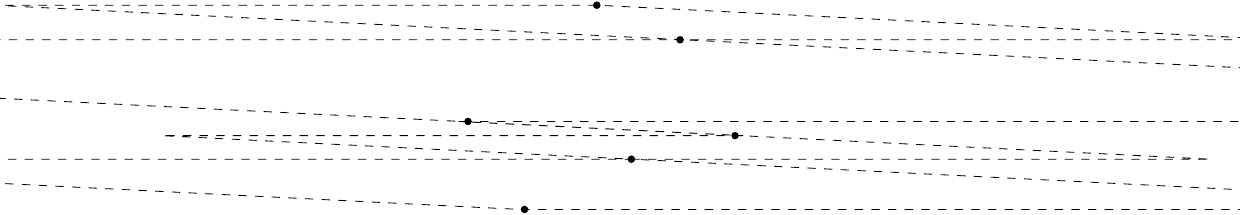}
    \caption{The initial decreasing configuration.}
    \label{bhull:fig:final_config}
  \end{minipage}
\end{figure}

During the transition between initial configurations, we recognize
four types of events that modify the set of vertices and overlapping
regions of $\bhullp$. An \emph{insertion} (resp. \emph{deletion})
event occurs when a vertex joins (resp. leaves) a
$\bset$-staircase. At \emph{overlap} (resp. \emph{release}) events, an
overlapping region is created (resp. destroyed).

Note that a vertex leaves (resp. joins) the same $\bset$-staircase at
most once, and thus, there is in total a linear number of insertion
(resp. deletion) events.~From Observation \ref{bhull:obs:staircases},
between $\beta_I$ and $\beta_D$ there is always an interval
$\phi = \left[ \beta_{1}, \beta_{2} \right]$ such that, for any
$\beta \in \phi$, the $\bhull$ of $P$ contains no overlapping
regions. Let us consider the angular intervals
$\phi_I = [\beta_I,\beta_{N_1}]$ and $\phi_D = [\beta_{N_2},\beta_D]$.
An angular sweep in $\phi_I$ results in a linear number of releasing
events caused by the deletion of all overlapping regions present at
the initial increasing configuration. As any vertex supports at most
two maximal $\bset$-quadrants, an additional linear number of region
events are generated by vertex events and, therefore, region events in
$\phi_I$ add up to $O(n)$. Using the same argument on $\phi_D$, we can
count a linear number of these events during an angular sweep.

\begin{lemma}\label{bhull:lemma:linear_events}
  There are $O(n)$ events during an angular sweep.
\end{lemma}

We now show how to compute the sequence of increasing angles that mark
vertex and overlapping events during an angular sweep.

\paragraph{Insertion and deletion events.}

The set of vertices of $\bhullp$ on the top-right $\bset$-staircase
has a total ordering that, at any value of $\beta$ is given by
traversing the staircase along its direction. At the
initial configuration, the order is also given by the sequence
$p_1,\ldots,p_n$ of points in $P$ labeled in ascending vertical
order.

Let us consider the set
$\alpha(P) = \{ \alpha_1, \ldots, \alpha_{n-1} \}$ where for each
$\alpha_i$, the slope of the line through $p_i$ and $p_{i+1}$ equals
$\tan(\alpha_i)$. In an increasing sweep, the first point leaving
the top-right $\bset$-staircase is $p_i$.
Indeed, for any $\beta > \alpha_i$, a
top-right $\bset$-quadrant with apex over $p_i$ is not $P$-free. This
is not the case for points corresponding to any $\alpha_j$ such that
$\alpha_j > \alpha_i$ and $\alpha_j > \beta$. See \Cref{bhull:fig:events}.

\begin{figure}[ht]
  \centering
  \begin{minipage}{0.9\textwidth}
    \centering
    \includegraphics[clip=true,trim=0 0.4cm 0 0, scale=1.2]{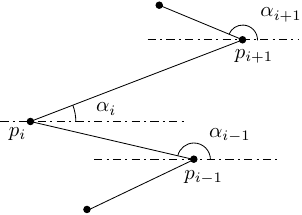}
    \caption{Insertion and deletion events for the top-right
      $\bset$-staircase.}
    \label{bhull:fig:events}
  \end{minipage}
\end{figure}

To compute the next value of $\beta$ where a point will leave the
top-right $\bset$-staircase, we must remove $\alpha_i$ from
$\alpha(P)$, update $\alpha_{i-1}$ to the angle where the slope of the
line through $p_{i-1}$ and $p_{i+1}$ equals $\tan(\alpha_{i-1})$,
and compute the new smallest element of $\alpha(P)$. A recursive
repetition of this computation allows us to obtain all deletion events
corresponding to the top-right $\bset$-staircase.

\begin{lemma}\label{bhull:lemma:point_events}
  All insertion and deletion events can be computed in $O(n\log n)$
  time and $O(n)$ space.
\end{lemma}
\begin{proof}
  Store the points in $P$ in a balanced search tree ordered according
  to the $y$-axis, and the set $\alpha(P)$ in a priority queue. From
  \Cref{bhull:lemma:linear_events}, the algorithm described above
  requires $O(n \log n)$ time and $O(n)$ space to compute the sets of
  insertion and deletion events, associated with the top-right
  $\bset$-staircase.~Considering the angles shown in
  \Cref{bhull:fig:point_events}, a similar algorithm can be used to
  obtain the corresponding events for the remaining $\bset$-staircases
  in the same time and space complexity.
\end{proof}

\begin{figure}[ht]
  \centering
  \begin{minipage}{0.9\textwidth}
    \centering
    \includegraphics[clip=true,trim=0 0.4cm 0 0, scale=1.2]{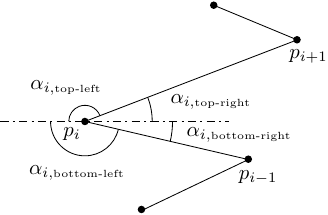}
    \caption{\Cref{bhull:lemma:point_events}.}
    \label{bhull:fig:point_events}
  \end{minipage}
\end{figure}

\paragraph{Overlap and release events.}

Let $Q_{r}$ and $Q_{l}$ be respectively, a pair of overlapping
top-right and bottom-left maximal $\bset$-quadrants. Consider that
$Q_r$ is supported by the vertices $p_j,p_{j+1}$, and $Q_l$ by the
vertices $p_k,p_{k+1}$. Also, assume the supporting points are labeled
according to the total ordering of their corresponding staircases
(see Figure~\ref{bhull:fig:eventsw}).

\begin{figure}[ht]
  \centering
  \begin{minipage}{0.9\textwidth}
    \centering
    \includegraphics[clip=true,trim=0 0 0 0, scale=1.2]{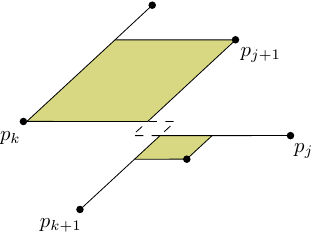}
    \caption{An overlapping region (bounded by dashed lines) generated
      by the intersection between a top-right and a bottom-left
      maximal $\bset$-quadrants.}
    \label{bhull:fig:eventsw}
  \end{minipage}
\end{figure}

The \emph{full overlap event} for the overlapping region defined by
$Q_{r}$ and $Q_{l}$ is the angle $\omega$ for which the slope of the
line through $p_{j+1}$ and $p_{k+1}$ equals $\tan(\omega)$. If
the supporting points do not leave their corresponding staircases,
this event marks the value of $\beta$ where the overlapping region
disappears.

Let $\omega(P)$ be the set of full overlap events for all the
overlapping regions at the initial increasing configuration, and
$\alpha_d(P)$ the set of all deletion events corresponding to the
vertices over the top-right and bottom-left $\bset$-staircases. Let
$\omega_{m}$ and $\alpha_m$ be the smallest values in $\omega(P)$ and
$\alpha_d(P)$, respectively. Performing an increasing sweep, to obtain
the first release event, we need to deal with the following cases:
\begin{enumerate}
\item \label{bhull:step_1} $\alpha_{m}$ corresponds to a supporting
  point, and $\alpha_{m} \leq \omega_{m}$.\, In this case,
  $\alpha_{m}$ needs to be processed and $\omega(P)$ needs to be
  updated. By removing a supporting point, at most two overlapping
  regions are terminated (two release events are added to
  $\omega(P)$), and at most one new overlapping region is generated
  (one overlapping event and one full overlap event are added to
  $\omega(P)$). After updating $\omega(P)$, $\omega_{m}$ and
  $\alpha_{m}$ are recomputed and the test is repeated.
\item $\alpha_{m}$ does not correspond to a supporting point. In this
  case, $\omega_{m}$ is the first release event.
\end{enumerate}

To compute the next release event, we must remove the current release
event from $\omega(P)$, and recompute $\omega_m$ as described above. A
recursive repetition of these steps allow us to obtain all release
events caused by intersections between the top-right and bottom-left
$\bset$-staircases.

\begin{lemma}\label{bhull:lemma:overlap_events}
  All overlap and release events can be computed in $O(n \log n)$ time
  and $O(n)$ space.
\end{lemma}
\begin{proof}
  Store the points in $P$ in a balanced search tree ordered according
  to the $y$-axis, and the sets $\alpha_d(P),\omega(P)$ in priority
  queues. From \Cref{bhull:lemma:linear_events}, the algorithm
  described above requires $O(n \log n)$ time and $O(n)$ space to
  compute the sets of overlap and release events associated with the
  top-right and bottom-left $\bset$-staircases.~A similar algorithm
  can be used to obtain the events associated to the top-right and
  bottom-left $\mathcal{O}_{\beta}$-staircases, with the same time and
  space upper bounds.
\end{proof}

\paragraph{Maintaining $\bhullp$.}

Considering the previous results, the maintenance of $\bhullp$ is
straightforward:
\begin{enumerate}
\item \label{bhull:maintain:step_1} Compute all vertex and overlap
  events, and store them in a list sorted by appearance during an
  increasing sweep.
\item \label{bhull:maintain:step_2} Compute
  $\bhullp[P][\beta_I]$. Store in height balanced trees the total
  orders of the sets of vertices lying over the four
  $\bset$-staircases. Store the set of overlapping regions in any
  constant-time access data structure (such as a hash table).
\item \label{bhull:maintain:step_3} Simulate the angular sweep by
  traversing the list of events.~At each insertion and deletion event,
  update the corresponding set of vertices.~At each overlap and
  release event, update the set of overlapping regions.
\end{enumerate}

From \Cref{bhull:lemma:point_events,bhull:lemma:overlap_events}, to
compute the sets of vertex and overlap events, we require
$O(n \log n)$ time and $O(n)$ space.
As we have a linear number of elements on each set,
we can merge them into a single ordered set using
$O(n \log n)$ time. Thus,
\cref{bhull:maintain:step_1} requires $O(n \log n)$ time and $O(n)$
space.

From \Cref{intro:thm:fixed_computation}, computing $\bhullp$ for any
fixed value of $\beta$ takes $O(n \log n)$ time and $O(n)$
space. Every $\bset$-staircase contains at most $n$ elements and
therefore, to store their total order in a height balanced tree we
require $O(n \log n)$ time.
Using a hash table, we can initialize the set of overlapping regions
in $O(n)$ time. Therefore,
\cref{bhull:maintain:step_2} requires $O(n \log n)$ time and $O(n)$
space.

At each insertion and deletion event, updating the corresponding set
of $\bset$-maximal elements requires $O(\log n)$ time per
operation. Updates on the set of overlapping regions takes constant
time, so \cref{bhull:maintain:step_3} takes $O(n \log n)$ time.~From
this analysis we get that, in total, we can compute and maintain
$\bhullp$ through an angular sweep in $O(n \log n)$ time and $O(n)$
space.  From \Cref{intro:thm:fixed_computation}, this time complexity
is optimal.

\begin{theorem}
  Computing and maintaining $\bhullp$ through an angular sweep
  requires $\Theta(n \log n)$ time and $O(n)$ space.
\end{theorem}

\section{Application problems}\label{sec:applications}

In this section we extend the results from Section \ref{sec:bhull} to
the solution of related optimization problems. We deal with the
problem of maximizing the area and the perimeter of $\bhullp$ (Sections~\ref{sec:apps:area} and~\ref{sec:apps:perimeter}, respectively). As an extra application, in~\ref{sec:apps:fitting} we deal with the
problem of fitting a two-joint polygonal chain to a point set.

\subsection{Area optimization.}\label{sec:apps:area}

In this section we solve the following problem:

\begin{problem}[Maximum area]
  Given a set $P$ of $n$ points in the plane, compute the value of
  $\beta$ for which $\bhullp$ has maximum area.
\end{problem}

Let $\{ \beta_1, \ldots, \beta_{O(n)} \}$ be the sequence of (vertex
and overlapping) events, ordered by appearance during an increasing
sweep. Following the lines of~\citet{bae_2009} (see also \Cref{apps:area:fig:area}), we express the area of
$\bhullp$ for any $\beta \in [\beta_i,\beta_{i+1})$ as
\begin{equation}
  \label{apps:area:eqn:area}
  \area(\bhullp) = \area(\polygon)
  - \sum_i \area(\triangles)
  + \sum_j \area(\parallelograms),
\end{equation}
where
 $\mathcal{P}(\beta)$ denotes the (simple) polygon having the same vertices as $\mathcal{O}_{\beta}\mathcal{H}(P)$ and an edge connecting two vertices if they are
consecutive in a $\mathcal{O}_{\beta}$-staircase.~The term $\triangle_i(\beta)$ is the $i$-th triangle defined by two consecutive vertices in a $\mathcal{O}_{\beta}$-staircase, and $\parallelogram_j(\beta)$ is the $j$-th overlapping region defined by the intersection of two opposite $\mathcal{O}_{\beta}$-staircases.

\begin{figure}[ht]
  \centering
  \begin{minipage}{0.9\textwidth}
    \centering
    \includegraphics[scale=1.5]{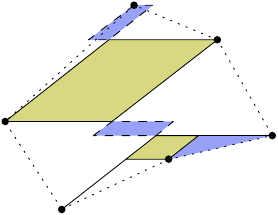}
    \caption{The area of $\bhullp$. The polygon $\polygon$ is bounded
      by dotted lines. A triangle $\triangles$ and two parallelograms
      $\parallelograms$ are filled in blue.}
    \label{apps:area:fig:area}
  \end{minipage}
\end{figure}

Our general approach is to maintain the terms of
\Cref{apps:area:eqn:area} during a complete angular sweep. We first
compute the optimal value of $\beta$ for $[\beta_1,\beta_2)$. We then
traverse the event sequence, updating the affected terms in
\Cref{apps:area:eqn:area} at each event. At the same time, we compute
the local angle of maximum area for each $[\beta_i,\beta_{i+1})$. With
any new computation, we keep the local optimal angle only if the
previous maximum area is improved.

\paragraph{The polygon $\polygon$.}

At any fixed value of $\beta$, the polygon can be constructed from the
vertices of $\bhullp$ in linear time. Once constructed, it takes a
second linear run to compute its area. During an interval between
events the area does not change. As $\polygon$ only depends on the
vertices of $\bhullp$, it is only modified by insertion and deletion
events. Each event can be handled in constant time: the area of a
triangle needs to be added (deletion event) or subtracted (insertion
event) from the previous value of the area of $\polygon$. See
\Cref{apps:area:fig:polygon}.

\begin{figure}[ht]
  \centering
  \subcaptionbox{\label{apps:area:fig:polygon:1}}
  {\includegraphics[scale=1.2]{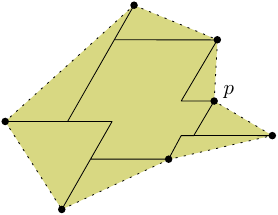}}
  \hspace{1.5cm}
  \subcaptionbox{\label{apps:area:fig:polygon:2}}
  {\includegraphics[scale=1.2]{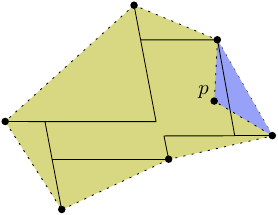}}
  \caption{Updating
    $\area(\polygon)$. \subref{apps:area:fig:polygon:1} The vertex $p$
    will leave the top-right $\bset$-staircase in an increasing sweep.
    \subref{apps:area:fig:polygon:2} The area of a triangle needs to
    be added after the deletion event from $\area(\polygon)$, once $p$ is no longer a
    vertex.}
  \label{apps:area:fig:polygon}
\end{figure}

\paragraph{The triangles $\triangles$.}

A triangle is defined by a pair of consecutive vertices of $\polygon$. If we
consider a top-right $\bset$-staircase, the area of $\triangles$ is
bounded by a line through $p_i$ and $p_{i+1}$, an horizontal line
through $p_i$, and a line with slope $\tan(\beta)$ through
$p_{i+1}$. In this context, the area of $\triangles$ is given by

\begin{align}
  \label{apps:area:eqn:triangle_area}
  \area(\triangles) &= \left| (x_i-x_{i+1})(y_{i+1}-y_i)+(y_{i+1}-y_i)^2\cot(\beta) \right| \nonumber \\
                    &= \left| a_i \pm b_i \cot(\beta) \right|,
\end{align}
with $a_i, b_i$ constants, where $\left( x_i,y_i \right)$ and $\left( x_{i+1},y_{i+1} \right)$
are respectively, the coordinates of the points $p_i$ and $p_{i+1}$.

\begin{figure}[ht]
  \centering
  \subcaptionbox{\label{apps:area:fig:triangle:1}}
  {\includegraphics[scale=1.2]{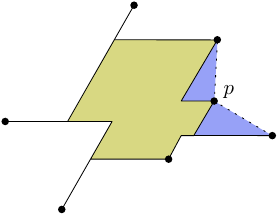}}
  \hspace{1cm}
  \subcaptionbox{\label{apps:area:fig:triangle:2}}
  {\includegraphics[scale=1.2]{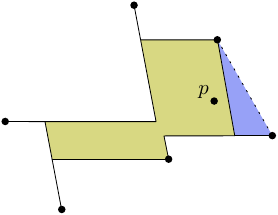}}
  \caption{Updating the term $\sum_i \area(\triangles)$.
    \subref{apps:area:fig:triangle:1} The point $p$ will leave the
    top-right $\bset$-staircase during an increasing
    sweep. \subref{apps:area:fig:polygon:2} When $p$ is no longer a
    vertex, two triangles are deleted, and a new triangle is created.}
  \label{apps:area:fig:triangle}
\end{figure}

The term $\sum_i \area(\triangles)$ is impacted by insertion and
deletion events and, at each event, it needs to be modified a constant
number of times. As any vertex of $\bhullp$ supports at most two
maximal $\bset$-quadrants, at a deletion event two triangles are
removed and one triangle is added. The converse occurs for insertion
events. See \Cref{apps:area:fig:triangle}.

\paragraph{The overlapping regions $\parallelograms$.}

An overlapping region is defined by two pairs of consecutive vertices
of $\bhullp$ belonging to opposite $\bset$-staircases. Overlapping
regions are bounded by parallelograms with sides parallel to the lines
in $\bset$. If we consider top-right and bottom-left
$\bset$-staircases intersecting as shown in
\Cref{apps:area:fig:parallelogram}, the area of a parallelogram is
given by
\begin{align}
  \label{apps:area:eqn:or_area}
  \area(\parallelograms) &= \left| (x_{k+1} - x_{i+1})(y_{k} - y_{i})
                           + (y_{k+1} - y_{i+1})(y_{k} - y_{i})\cot(\beta) \right| \nonumber \\
                         &= \left| a_j \pm b_j \cot(\beta) \right|,
\end{align}
with $a_i, b_i$ constants, where $p_i=(x_i,y_i),p_{i+1}=(x_{i+1},y_{i+1})$ and
$p_k=(x_k,y_k),p_{k+1}=(x_{k+1},y_{k+1})$ are respectively, the
supporting vertices of the overlapping maximal opposite
$\bset$-quadrants.

\begin{figure}[ht]
  \centering
  \subcaptionbox{\label{apps:area:fig:parallelogram:1}}
  {\includegraphics[scale=1.2]{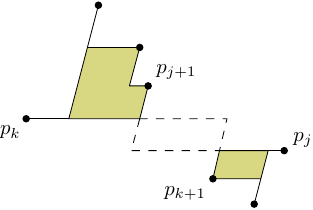}}
  \hspace{1.5cm}
  \subcaptionbox{\label{apps:area:fig:parallelogram:2}}
  {\includegraphics[scale=1.2]{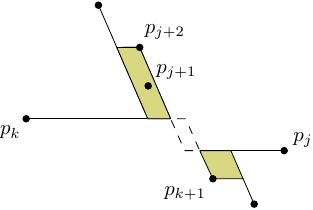}}
  \caption{An overlapping region destroyed because of the vertex
    $p_{j+1}$ leaving the top-right $\bset$-staircase, during an
    increasing sweep.}
  \label{apps:area:fig:parallelogram}
\end{figure}

The term $\sum_j \area(\parallelograms)$ is impacted by all types of
events. Overlap and release events require a single overlapping
region to be added or deleted. For insertion and deletion events, at
most two new overlaps are created, or destroyed.

\paragraph{Characterization.}

Before describing our algorithm, in the following lemmas we answer
some basic questions about the behavior of
$\area(\bhullp)$. \Cref{apps:area:lemma:max_angle,apps:area:lemma:bimodal}
imply that it seems not possible to restrict the number of
candidate angles of maximum area. On the other hand, Lemma
\ref{apps:area:lemma:events} shows that the angle of maximum area is
actually located at an event.

\begin{lemma}\label{apps:area:lemma:max_angle}
  For any $\beta_0 \in (0,\pi)$  there exists a point set $P$ such that \[\max_{\beta} \operatorname{area}(\mathcal{O}_\beta\mathcal{H}(P))\neq\operatorname{area}(\mathcal{O}_{\beta_0}\mathcal{H}(P)).\]
\end{lemma}
\begin{proof}
  Consider the coordinate system formed by $\bset[\beta_0]$. Place one
  point over the $y^+$-, $y^-$-, and $x^+$-semiaxes, and a point over
  the second quadrant (see \Cref{apps:area:fig:max_angle:1}). From
  this position, note that $\area(\bhullp) = 0$ for any
  $\beta \leq \beta_0$ (\Cref{apps:area:fig:max_angle:2}), and there
  exists at least one $\beta_1 > \beta_0$ such that
  $\area(\bhullp[P][\beta_1]) \neq 0$
  (\Cref{apps:area:fig:max_angle:3}). Hence $\beta_0$ cannot be the
  angle of maximum area.
\end{proof}

\begin{figure}[]ht
  \centering
  \subcaptionbox{\label{apps:area:fig:max_angle:1}}
  {\includegraphics{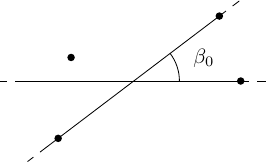}}
  \hspace{.5cm}
  \subcaptionbox{\label{apps:area:fig:max_angle:2}}
  {\includegraphics{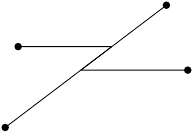}}
  \hspace{.5cm}
  \subcaptionbox{\label{apps:area:fig:max_angle:3}}
  {\includegraphics{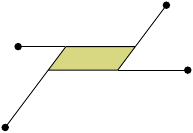}}
  \caption{Lemma
    \ref{apps:area:lemma:max_angle}. \subref{apps:area:fig:max_angle:1}
    The set of points. \subref{apps:area:fig:max_angle:2}
    $\area(\bhullp) = 0$ for $\beta \leq
    \beta_0$.
    \subref{apps:area:fig:max_angle:3} $\area(\bhullp) \neq 0$ for
    some $\beta > \beta_0$.}
  \label{apps:area:fig:max_angle}
\end{figure}

\begin{lemma}\label{apps:area:lemma:bimodal}
  For any $\beta_0,\beta_1 \in (0,\pi)$, there exists a point set $P$ for which $\area(\mathcal{O}_{\beta}(P))$ has local maxima in $\beta_0$ and $\beta_1$.
\end{lemma}

\begin{proof}
  Let $\ell_0$ be a line with slope $\tan(\beta_0)$, $\ell_1$ a line
  with slope $\tan(\beta_1)$, and without loss of generality, let us
  assume that $\beta_0 < \beta_1$. We define $p_l, p_r,p_t,$ and $p_c$
  to be the points located respectively, at the left corner, right
  corner, top corner, and the interior of the triangle bounded by the
  $x$-axis, $\ell_0$, and $\ell_1$. See
  \Cref{apps:area:fig:bimodal_1}.

  \begin{figure}[ht]
    \centering
    \begin{minipage}{0.9\textwidth}
      \centering
      {\includegraphics[scale=1.2]{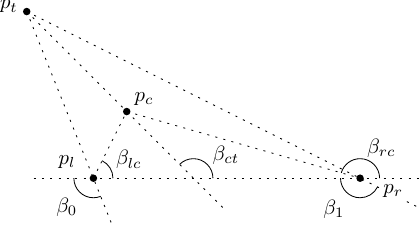}}
      \caption{The points configuration.}
      \label{apps:area:fig:bimodal_1}
    \end{minipage}
  \end{figure}

  Consider the angles $\beta_{lc}, \beta_{ct}$, and $\beta_{rc}$ as in
  \Cref{apps:area:fig:bimodal_1}. Note that
  $\beta_{lc} < \beta_{0} < \beta_{ct} < \beta_{1} < \beta_{rc}$.
  Using an increasing sweep from the initial increasing configuration the
  first release event is $\beta_{lc}$. From there, the area of
  $\bhullp$ is given by a parallelogram $\parallelogram_{lc}$ of
  constant height, so both the base of $\parallelogram_{lc}$ and the
  area of $\bhullp$ increase or decrease together as $\beta$
  changes. As $\beta$ goes from $\beta_{lc}$ to $\beta_{0}$, the base
  of $\parallelogram_{lc}$ increases up to $\beta_0$, there
  exist a local maximum. The base of $\parallelogram_{lc}$ then
  decreases from $\beta_0$ to $\beta_{ct}$, to increase again from
  $\beta_{ct}$ to $\beta_{1}$. At $\beta_{1}$ there is a second local
  maximum, as the base of $\parallelogram_{lc}$ starts decreasing
  again after $\beta_{1}$ up to the last construction event at
  $\beta_{rc}$, where the area of $\bhullp$ is zero. See
  \Cref{apps:area:fig:bimodal_2}.
\end{proof}

 \begin{figure}[ht]
    \centering
    \subcaptionbox{\label{apps:area:fig:bimodal_2:1}}
    {\includegraphics[scale=0.65]{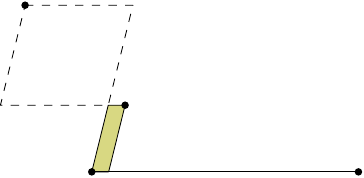}}
\hspace{-0.2cm}
    \subcaptionbox{\label{apps:area:fig:bimodal_2:2}}
    {\includegraphics[scale=0.65]{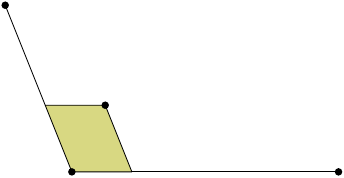}}
\hspace{-0.2cm}
    \subcaptionbox{\label{apps:area:fig:bimodal_2:3}}
    {\includegraphics[scale=0.65]{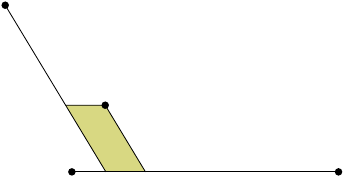}}
    \hspace{-0.2cm}
    \subcaptionbox{\label{apps:area:fig:bimodal_2:4}}
    {\includegraphics[scale=0.65]{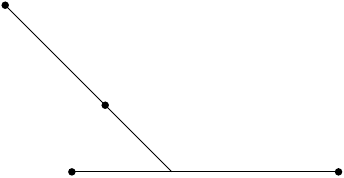}}
\hspace{-0.2cm}
    \subcaptionbox{\label{apps:area:fig:bimodal_2:5}}
    {\includegraphics[scale=0.65]{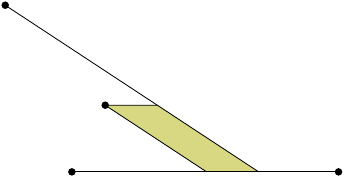}}
\hspace{-0.2cm}
    \subcaptionbox{\label{apps:area:fig:bimodal_2:6}}
    {\includegraphics[scale=0.65]{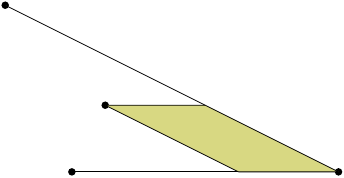}}
\hspace{-0.2cm}
    \subcaptionbox{\label{apps:area:fig:bimodal_2:7}}
    {\includegraphics[scale=0.7]{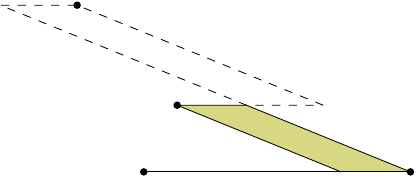}}
    \caption{Increasing sweep over the point set of
      \Cref{apps:area:fig:bimodal_1}. \subref{apps:area:fig:bimodal_2:1}
      $\beta=\beta_0-\epsilon$. \subref{apps:area:fig:bimodal_2:2} A
      local maximum on
      $\beta=\beta_0$. \subref{apps:area:fig:bimodal_2:3}
      $\beta\in(\beta_0,\beta_{ct})$. \subref{apps:area:fig:bimodal_2:4}
      A local minimum on
      $\beta=\beta_{ct}$. \subref{apps:area:fig:bimodal_2:5}
      $\beta\in(\beta_{ct},\beta_1)$. \subref{apps:area:fig:bimodal_2:6}
      A second local maximum on
      $\beta=\beta_1$. \subref{apps:area:fig:bimodal_2:7}
      $\beta=\beta_1+\epsilon$.}
    \label{apps:area:fig:bimodal_2}
  \end{figure}


\begin{lemma}\label{apps:area:lemma:events}
  The area of $\bhullp$ reaches its maximum at values of $\beta$
  belonging to the sequence of events.
\end{lemma}

\begin{proof}
  Let us consider the area of $\bhullp$ given by
  \Cref{apps:area:eqn:area}. From
  \Cref{apps:area:eqn:triangle_area,apps:area:eqn:or_area}, the area
  of $\bhullp$ can be rewritten as

  \begin{align}
    \label{apps:area:eqn:event_1}
    \area(\bhullp) &= \area(\polygon)
                     - \sum_i \area(\triangles)
                     + \sum_j \area(\parallelograms) \nonumber \\
                   &= \area(\polygon)
                     - \sum_i \left| a_i \pm b_i \cot(\beta) \right|
                     + \sum_j \left| a_j \pm b_j \cot(\beta) \right|.
  \end{align}
  If we consider the different point configurations that define a
  triangle (see \Cref{apps:area:fig:events}), we can express
  $|a_i \pm b_i \cot(\beta)|$ as $a_i + b_i \cot(\beta)$ or
  $a_i - b_i \cot(\beta)$, according to the specific
  configuration. Thus, we have
    \begin{align*}
    \sum_i \area(\triangles) &= \sum_i |a_i \pm b_i\cot(\beta)| \\
                             &= \sum_{i_0} \left( a_{i_0} + b_{i_0}\cot(\beta) \right)
                               + \sum_{i_1} \left( a_{i_1} - b_{i_1}\cot(\beta)\right)= a + b\cot(\beta).
  \end{align*}

  \begin{figure}[ht]
    \centering
    \subcaptionbox{\label{apps:area:fig:events:1}}
    {\includegraphics{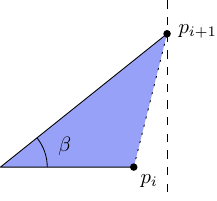}}
    \hspace{1.5cm}
    \subcaptionbox{\label{apps:area:fig:events:2}}
    {\includegraphics{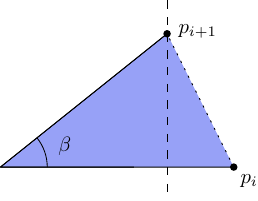}}
    \caption{Relative positions between the vertices of the triangle
      $\triangles$.}
    \label{apps:area:fig:events}
  \end{figure}

  It is possible to make a similar case-by-case analysis for the
  overlapping regions, to obtain from \Cref{apps:area:eqn:or_area} an
  expression with the form $c + d\cot(\beta)$. Within an interval
  between events $P$ does not change, and its area remains
  constant. Therefore, in an interval $[\beta_i,\beta_{i+1})$ we can rewrite:

  \begin{align}
    \label{apps:area:eqn:event_2}
    \area(\bhullp) &= \area(\polygon)
                     - \sum_i |a_j \pm b_j \cot(\beta)|
                     + \sum_j |a_j \pm b_j \cot(\beta)| \\
                   &= \area(\polygon)
                     - \left( a + b\cot(\beta) \right)
                     + \left( c + d\cot(\beta) \right) \nonumber \\
                   &= \area(\polygon)
                     + (c-a) +(d-b)\cot(\beta) \nonumber \\
    \label{apps:area:eqn:event_3}
                   &= A + B\cot(\beta),
  \end{align}
  where $A$ and $B$ contain the sum of all constants from the terms in
  \Cref{apps:area:eqn:event_2}. Note that \Cref{apps:area:eqn:event_3}
  is monotone at any interval $[\beta_i,\beta_{i+1})$, as it is
  monotone in $(0,\pi)$. Depending on the particular values of $A$ and
  $B$, $\area(\bhullp)$ might be non-decreasing or
  non-increasing. Thus, the local maximum is given either by $\beta_i$
  or $\beta_{i+1}$.
\end{proof}

\paragraph{The search algorithm.}

The algorithm to compute the angle of optimum area is outlined as
follows.
\begin{enumerate}
\item \label{apps:area:step_1}Traverse the sequence of events to
  identify the first release event $\beta_d$, and the last
  overlap event $\beta_c$.
  Restrict the sequence to start with $\beta_d$ and finish with $\beta_c$, so that $\mathcal{O}_{\beta}(P)$ has at least one connected component in every interval. Ignored events have no effect in the
  result, as they belong to an initial (increasing or decreasing) configuration, where $\operatorname{area}(\mathcal{O}_\beta\mathcal{H}(P))=0$.
\item \label{apps:area:step_2} At the first interval, compute
$\mathcal{O}_\beta\mathcal{H}(P)$ and using \Cref{apps:area:eqn:area} compute
$\operatorname{area}(\mathcal{O}_\beta\mathcal{H}(P))$, keeping the angle
$\beta_m$ of maximum area.
\item \label{apps:area:step_3}Traverse the sequence of events. At each
  event:
  \begin{enumerate}
  \item \label{apps:area:step_3_1}Update the set of vertices and
    overlapping regions of $\bhullp$.
  \item \label{apps:area:step_3_2}Handle each event updating
    \Cref{apps:area:eqn:area} as explained above.
  \item \label{apps:area:step_3_3}Compute the local angle of maximum
    area. Replace $\beta_m$ only if the area of $\bhullp$ is improved.
  \end{enumerate}
\end{enumerate}

There is a linear number of events in total, so
step~\ref{apps:area:step_1} requires $O(n)$
time. \Cref{apps:area:eqn:area} contains at most a linear number of terms, as
there is at most a linear number of vertices and overlapping
regions. Thus, from \Cref{intro:thm:fixed_computation} and previous
discussions, step~\ref{apps:area:step_2} requires $\Theta(n \log n)$ time
and $O(n)$ space.

From Section~\ref{sec:bhull:sweep}, the updates on step~\ref{apps:area:step_3_1}
require logarithmic time. Every event results in a constant number of modifications
to~\Cref{apps:area:eqn:area}, as we described previously in this
section. From~\Cref{apps:area:lemma:events} we can obtain the angle of
maximum area in constant time. As there is a linear number of events,
step~\ref{apps:area:step_3} requires a total of $O(n\log n)$ time. From
this analysis we obtain the following Theorem, where the lower bound comes
from the maintenance of $\bhullp$.

\begin{theorem}
  Computing the value(s) of $\beta \in (0, \pi)$ for which $\bhullp$
  has maximum area, requires $\Theta(n \log n)$ time and $O(n)$ space.
\end{theorem}

\subsection{Perimeter optimization.}\label{sec:apps:perimeter}

In this section we solve the following problem:

\begin{problem}[Maximum perimeter]
  Given a set $P$ of $n$ points in the plane, compute the value of
  $\beta$ for which $\bhullp$ has maximum perimeter.
\end{problem}

\noindent The perimeter of $\bhullp$ is given by
\begin{equation}
  \label{apps:perim:eqn:perim}
  \perim(\bhullp) = \sum_i \perim(\angle_i(\beta))
  - \sum_j \perim(\parallelograms)
  - \sum_k \perim(\diagdown_k(\beta)),
\end{equation}
where the $\angle_i(\beta)$ and the $\parallelograms$ denote the \emph{steps}
 and parallelograms, respectively, defined
by the staircases, and $\diagdown_k$ denotes one of the (at most four)
\emph{antennas} of $\bhullp$, that is, a segment of an
$\bset$-staircase bounding a zero-area region of $\bhullp$.
See again~\Cref{apps:area:fig:area}.

The same approach, and most of the arguments we used to maximize the
area can be applied here. Following the same ideas, we will first
analyze the computation and maintenance of~\Cref{apps:perim:eqn:perim},
we then present adaptations of~\crefrange{apps:area:lemma:max_angle}{apps:area:lemma:events}, and
finalize outlining the search algorithm.

\paragraph{The steps $\angle_i(\beta)$.}

Considering a top-right $\bset$-staircase (see again
\Cref{apps:area:fig:triangle}), the perimeter of $\angle_i(\beta)$ is given
by \Cref{apps:perim:eqn:step}, where $p_i=(x_i,y_i)$ and
$p_{i+1}=(x_{i+1},y_{i+1})$ are the points supporting the $i$-th
step. Vertices over the staircase have non-decreasing $y$ coordinates,
so $a_i$ is always positive. Event handling is done in the same way as
we did with triangles in the previous section.

\begin{align}
  \label{apps:perim:eqn:step}
  \perim(\angle_i(\beta)) &= \left|
                      \left( y_{i+1}-y_i \right)\cot(\beta)
                      + \left( y_{i+1}-y_i \right)\csc(\beta)
                      + \left( x_i-x_{i+1} \right)
                      \right| \nonumber \\
                    &= \left|
                      a_i \left( \cot(\beta) + \csc(\beta) \right)
                      \pm  b_i
                      \right|.
\end{align}

\paragraph{The overlapping regions $\parallelograms$.}

If we consider top-right and bottom-left $\bset$-stair-cases
intersecting as shown in \Cref{apps:area:fig:parallelogram}, the
perimeter of an overlapping region is given by
\Cref{apps:perim:eqn:or}. The constants $c_j$ and $d_j$ are always
positive. Event handling is done in the same way as we handled
overlapping regions to optimize the area of $\bhullp$.

\begin{align}
  \label{apps:perim:eqn:or}
  \perim(\parallelograms) &= \left| 2(y_{i+1} - y_{k+1}) \cot(\beta)
                           + 2(y_{k} - y_{i})\csc(\beta)
                           - (x_{i+1} - x_{k+1}) \right| \nonumber \\
                         &= \left| c_j \cot(\beta)
                           + d_j \csc(\beta)
                           \pm e_j \right|.
\end{align}

\paragraph{The antennas $\diagdown_k(\beta)$.}

An antenna is a semistep at one of the extremes of an
$\bset$-staircase. Just as steps and triangles, an antenna is defined
by two consecutive $\bset$-maximal points. If we consider a top-right
$\bset$-staircase, the perimeter of an antenna is given by
\Cref{apps:perim:eqn:antenna:1} if it is the first semistep of the
staircase, and by \Cref{apps:perim:eqn:antenna:2} if it is the last one
(see \Cref{apps:perim:fig:antennas}). In both equations we consider
$p_i=(x_i,y_i)$ to be the point supporting the corresponding
semistep. The constant $f_k$ is always positive.

\begin{align}
  \label{apps:perim:eqn:antenna:1}
  \perim_f(\diagdown_k) &= \left|
                       \left( y_{i+1}-y_i \right)\cot(\beta)
                       + \left( x_i-x_{i+1} \right)
                       \right| \nonumber \\
                     &= \left|
                       f_k \cot(\beta) \pm g_k
                       \right|\\[1em]
  \label{apps:perim:eqn:antenna:2}
  \perim_l(\diagdown_k) &= \left( y_{i+1}-y_i \right)\csc(\beta) \nonumber \\
                     &= f_k \cot(\beta)
\end{align}

\begin{figure}[ht]
  \centering
  \begin{minipage}{0.9\textwidth}
    \centering
    {\includegraphics{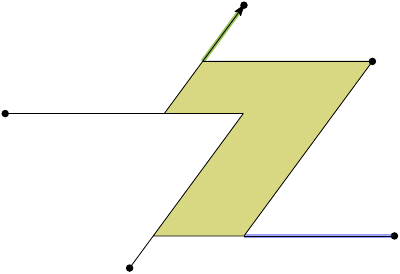}}
    \caption{Two antennas respectively, at the first (horizontal
      segment) and last (non-horizontal segment) semisteps of the
      top-right $\bset$-staircase.}
    \label{apps:perim:fig:antennas}
  \end{minipage}
\end{figure}

Considering the case-by-case analysis we did in the previous section,
we can rewrite Equations~\ref{apps:perim:eqn:step} to
\ref{apps:perim:eqn:antenna:2} as

\begin{align}
  \label{apps:perim:eqn:antenna:3}
  \sum_i \perim(\angle_i(\beta)) &= \sum_i \left|
                                   a_i \cot(\beta) + a_i \csc(\beta)
                                   \pm  b_i
                                   \right| \nonumber \\
                                 &= \sum_{i_0}
                                   a_{i_0} \cot(\beta) + a_{i_0} \csc(\beta) + b_{i_0}
                                   +  \sum_{i_1}
                                   a_{i_1} \cot(\beta) + a_{i_1} \csc(\beta) - b_{i_1} \nonumber \\
                                 &= a \cot(\beta) + a \csc(\beta) + b
\end{align}

\begin{align}
  \label{apps:perim:eqn:antenna:4}
  \sum_j \perim(\parallelograms) &= \sum_j \left|
                                   c_j \cot(\beta)
                                   + d_j \csc(\beta)
                                   \pm e_j \right| \nonumber \\
                                 &= \sum_{j_0}
                                   c_{j_0} \cot(\beta) + d_{j_0} \csc(\beta) + e_{j_0}
                                   +  \sum_{j_1}
                                   c_{j_1} \cot(\beta) + d_{j_1} \csc(\beta) - e_{j_1} \nonumber \\
                                 &= c \cot(\beta) + d \csc(\beta) + e
\end{align}

\begin{align}
  \label{apps:perim:eqn:antenna:5}
  \sum_k \perim_l(\diagdown_k) &= \left|
                                 f_k \cot(\beta) \pm g_k
                                 \right| \nonumber \\
                               &= \sum_{k_0}
                                 f_{k_0} \cot(\beta) + g_{k_0} + \sum_{k_1}
                                 f_{k_1} \cot(\beta) - g_{k_1} \nonumber \\
                               &= f \cot(\beta) + g,
 \end{align}

 and use
 \Crefrange{apps:perim:eqn:antenna:3}{apps:perim:eqn:antenna:5} to
 rewrite \Cref{apps:perim:eqn:perim} as

 \begin{align}
   \label{apps:perim:eqn:antenna:6}
   \perim(\bhullp) &= \sum_i \perim(\angle_i(\beta))
                     - \sum_j \perim(\parallelograms)
                     - \sum_k \perim(\diagdown_k(\beta)) \nonumber \\
                   &= \left( a + c + f \right) \cot(\beta)
                     + \left( a + d \right) \csc(\beta)
                     + \left( b + e + g \right) \nonumber \\
                   &= A \cot(\beta) + B \csc(\beta) + C.
 \end{align}

 Note that all the constants in \Cref{apps:perim:eqn:antenna:6} adding
 up to $A$ and $B$ are always positive, so $A,B > 0$. Moreover, within
 an interval there are at most four antennas, as they contain one of
 the left-most, right-most, top-most, and bottom-most points in $P$
 (see again \Cref{apps:perim:fig:antennas}). Therefore, the number of
 terms contributed by antennas to \Cref{apps:perim:eqn:perim} is
 constant and, except for $C$, they do not modify the original signs
 of any other term.

 For simplicity, we will avoid antennas in the optimization of the
 perimeter of $\bhullp$, by using a version of
 \Cref{apps:perim:eqn:perim} not containing the term
 $\sum_k \perim(\diagdown_k(\beta))$. From the discussion above, both
 expressions have maxima at the same values of $\beta$.

\paragraph{Characterization.}

We next answer questions about the behavior of $\perim(\bhullp)$,
similar to the ones answered with
\crefrange{apps:area:lemma:max_angle}{apps:area:lemma:events} in
\Cref{sec:applications}. Specifically, we show that the angle of
maximum perimeter corresponds to an event
(\Cref{apps:perim:lemma:events}) and, other than that, no restriction
on the candidate angles seems to be possible
(Lemmas~\ref{apps:perim:lemma:max_angle} and \ref{apps:perim:lemma:bimodal}).

\begin{lemma}\label{apps:perim:lemma:max_angle}
  For any $\beta_0 \in (0,\pi)$, there exists a point set such that
  \[\max_{\beta} \perim(\bhullp)\neq\perim(\bhullp[P][\beta_0]).\]
\end{lemma}

\begin{proof}
  Consider the coordinate system formed by $\bset[\beta_0]$. Place one
  point on the origin, and a point on the second and fourth
  quadrants (\Cref{apps:perim:fig:max_angle:1}). As the set of points
  is monotone with respect of the $x$- and $y$-axes,
  $\bhullp[P][\beta_0] = P$. Therefore, $\perim(\bhullp[P][\beta_0])$
  is equal to zero.

  \begin{figure}[ht]
    \centering \subcaptionbox{\label{apps:perim:fig:max_angle:1}}
    {\includegraphics[scale=1.2]{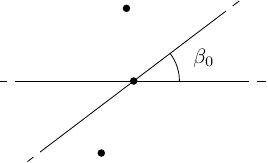}}
    \hspace{1cm}
    \subcaptionbox{\label{apps:perim:fig:max_angle:2}}
    {\includegraphics[scale=1.2]{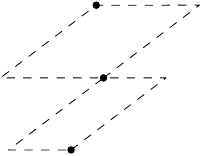}}
    \hspace{1cm}
    \subcaptionbox{\label{apps:perim:fig:max_angle:3}}
    {\includegraphics[scale=1.2]{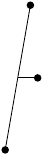}}
    \caption{Lemma
      \ref{apps:perim:lemma:max_angle}. \subref{apps:perim:fig:max_angle:1}
      A set of points. \subref{apps:perim:fig:max_angle:2}
      $\perim(\bhullp) = 0$ for $\beta \leq \beta_0$.
      \subref{apps:perim:fig:max_angle:3}
      $\perim(\bhullp[P][\beta_0]) \neq 0$ for some
      $\beta > \beta_0$.}
    \label{apps:perim:fig:max_angle}
  \end{figure}

  From this position, note that $\perim(\bhullp) = 0$ for any
  $\beta \leq \beta_0$ (\Cref{apps:perim:fig:max_angle:2}), and there
  exists at least one $\beta_1 > \beta_0$ such that
  $\perim(\bhullp[P][\beta_1]) \neq 0$
  (\Cref{apps:perim:fig:max_angle:3}). Clearly, $\beta_0$ is not the
  angle of maximum perimeter.
\end{proof}

\begin{lemma}\label{apps:perim:lemma:bimodal}
  For any $\beta_0,\beta_1 \in (0,\pi)$, there exists a point set $P$ for which $\area(\mathcal{O}_{\beta}(P))$ has local maxima in $\beta_0$ and $\beta_1$.
\end{lemma}

\begin{proof}
  \newcommand{\topt}{\triangle_t}
  \newcommand{\topal}{\beta_{tl}}
  \newcommand{\topar}{\beta_{tr}}
  \newcommand{\bott}{\triangle_b}
  \newcommand{\botal}{\beta_{bl}}
  \newcommand{\botar}{\beta_{br}}

  Let $\topt$ be an acute triangle bounded by the $x$-axis, and two
  lines $\ell_{tl},\ell_{tr}$ with slopes $\tan(\topal)$ and
  $\tan(\topar)$, respectively. Without loss of generality, we assume
  that $\topal < \topar$, and the intersection point between
  $\ell_{tl}$ and $\ell_{tr}$ lies on the $y^+$-semiplane.

  Let us consider the set $P' = \{ p_l,p_r,p_t \}$ of points located
  respectively, over the left, right, and top vertices of
  $\topt$. Note that, at any starting position, the perimeter of
  $\bhullp[P']$ is constant and equal to the base of $\topt$. Using an
  increasing sweep, from $\topal$ to $\topar$ the perimeter is formed
  additionally by a line segment $\ell_{t,b}$ joining $p_t$, and a
  point $p_b$ traversing the base of $\topt$ from $p_l$ to
  $p_r$. During this interval, both $\ell_{t,b}$ and the perimeter of
  $\bhullp[P']$ increase or decrease together as $\beta$ changes. On
  this conditions, the perimeter of $\bhullp[P']$ has a local minimum
  on $\beta = \frac{\pi}{2}$ and thus, a local maximum on $\topal$
  ($\perim(\bhullp[P'][\topal]) >
  \perim(\bhullp[P'][\topal-\varepsilon])$),
  and a second local maximum on $\topar$
  ($\perim(\bhullp[P'][\topar]) > \perim(\bhullp[P'][\topar +
  \varepsilon])$). See \Cref{apps:perim:fig:bimodal_1}.

  \begin{figure}[ht]
    \centering
    \subcaptionbox{\label{apps:perim:fig:bimodal_1:1}}
    {\includegraphics[scale=1.2]{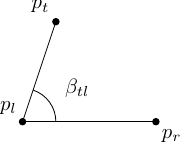}}
    \hspace{1cm}
    \subcaptionbox{\label{apps:perim:fig:bimodal_1:2}}
    {\includegraphics[scale=1.2]{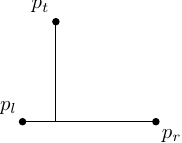}}
    \hspace{1cm}
    \subcaptionbox{\label{apps:perim:fig:bimodal_1:3}}
    {\includegraphics[scale=1.2]{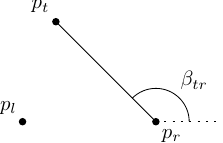}}
    \caption{\subref{apps:perim:fig:bimodal_1:1} and
      \subref{apps:perim:fig:bimodal_1:3} Maxima on $\beta_{tl}$ and
      $\beta_{tr}$. \subref{apps:perim:fig:bimodal_1:2} A minima on
      $\frac{\pi}{2}$.}
    \label{apps:perim:fig:bimodal_1}
  \end{figure}

  Let $\bott$ be a second acute triangle bounded by the $x$-axis, and
  a second pair of lines $\ell_{bl},\ell_{br}$ with slopes
  $\tan(\botal)$ and $\tan(\botar)$ that pass through $p_l$ and $p_r$,
  respectively. The angles are such that $\botal > \botar$, and the
  intersection point $p_b$ between $\ell_{bl}$ and $\ell_{br}$ lie on
  the $Y^-$ semiplane. Note that, if we add $p_b$ to the set $P'$, the
  arguments from the above discussion hold for both $\topt$ and
  $\bott$, so the perimeter of $\bhullp[P']$ has now local maxima on
  $\topal, \topar, \botal$, and
  $\botar$. See \Cref{apps:perim:fig:bimodal_2}.

  \begin{figure}[ht]
    \centering
    \begin{minipage}{0.9\textwidth}
      \centering
      \includegraphics[scale=1.2]{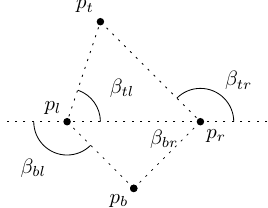}
      \caption{The angles with local maxima.}
      \label{apps:perim:fig:bimodal_2}
    \end{minipage}
  \end{figure}

  Given the angles $\beta_0$ and $\beta_1$, construct the previous
  point set as explained. In this construction,
  $\topal \leq \frac{\pi}{2} < \topar$ and
  $\botal > \frac{\pi}{2} \geq \botar$. Set two of
  $\topal, \topar, \botal, \botar$ to the values of $\beta_0$ and
  $\beta_1$ appropriately, according to the cases
  \begin{inparaenum}[\itshape i\upshape)]
  \item $\beta_0,\beta_1 < \frac{\pi}{2}$,
  \item $\beta_0,\beta_1 > \frac{\pi}{2}$, and
  \item $\beta_0 \leq \frac{\pi}{2} < \beta_1$ or
    viceversa\end{inparaenum}. The perimeter of $\bhullp[P']$ will
  have local maxima at $\beta_0$ and $\beta_1$.
\end{proof}

\begin{lemma}\label{apps:perim:lemma:events}
  The perimeter of $\bhullp$ reaches its maximum at values of $\beta$
  corresponding to sequence events.
\end{lemma}

\begin{proof}
  From \Cref{apps:perim:eqn:antenna:6} we know that the perimeter of
  $\bhullp$ is given by

  \begin{equation*}
    \perim(\bhullp) = A \cot(\beta) + B \csc(\beta) + C,
  \end{equation*}
  where $A,B \geq 0$. Looking for critical points in this expression,
  we arrive to
  \begin{equation}\label{apps:perim:eqn:derivative}
    \cos (\beta) = - \frac{A}{B},
  \end{equation}
  where $\beta \neq 0,\pi$.
  By analyzing the possible roots in \Cref{apps:perim:eqn:derivative},
  we deal with the following cases:
  \begin{enumerate}
  \item $A > B$. There are no roots in this case, as
    $\frac{A}{B} > 1$. The length of the perimeter is monotonic in an interval
    between events.
  \item $A = B$. There are again no roots in this case, as $\beta$
    cannot be $0$ nor $\pi$. The length of the perimeter is again monotonic in an
    interval between events.
  \item $A < B$. There is one root at
    $\beta = \cos^{-1}(-\frac{A}{B})$, as $A$ and $B$ are always
    positive and different from zero. In an interval between events we
    have one inflection point, so there are either two local maxima
    or two local minima, located at the endpoints of the interval.
  \end{enumerate}
\end{proof}

\paragraph{The search algorithm.}

We look for the maximum perimeter angles in the same way as we
obtained the values for maximum area. We first compute the list of
events, obtain the maximum perimeter angle for the first interval
between events, and repeat the procedure for the remaining
events. While traversing the event list, we update the optimum value
angle only if the previous value is improved. A similar complexity
analysis is also valid.

\begin{theorem}
  Computing the value(s) of $\beta \in (0, \pi)$ for which $\bhullp$
  has maximum perimeter takes $O(n\log n)$ time and $O(n)$
  space.
\end{theorem}

\section{Concluding remarks}\label{sec:conclusions}

We presented an algorithm to maintain the $\bhull$ of a planar point
set while $\beta$ runs from $0$ to $\pi$ and extended this result to
solve related optimization problems. We considered the maximization of
the area and the perimeter of $\bhullp$, and presented a variation of
the $2$-fitting problem studied in~\cite{fitting_2011}.~In our
version, the fitting curve is an alternating polygonal chain with
segments forming an angle $\beta$.

A natural extension of this work is to replace $\bset$ with a set
$\mathcal{O}$ containing more than two lines.~Different variations can
be obtained by restricting the orientations and (or) the number of
lines in $\mathcal{O}$. In particular, the characterization of the
area and perimeter functions on each variation, seems an interesting
and non-trivial problem.

As the Orthogonal Convex Hull, the $\bhull$ is suitable to be used as
a separator or an enclosing shape.~As it is always contained in the
standard convex hull (and therefore, in several other traditional
enclosing shapes), it is relevant in applications where the separator
or enclosing shape is required to have minimum area.~Finally, note
that we can easily extend the results from \Cref{sec:bhull} to
optimize the number of vertices of $\bhullp$, by keeping track of the
vertex count at each interval between events. Without much effort, the
approach and arguments from \citet{alegria_2013} can be extended to
$\bset$-convexity, and applied to problems related to containment
relations between $\bhull$s of colored point sets.

\bibliographystyle{plainnat}
\bibliography{references}

\appendix
\renewcommand\thesection{Appendix \Alph{section}}

\section{The oriented $\left( 2,\beta \right)$-fitting problem}\label{sec:apps:fitting}


For $k \geq 1$, $\theta \in [0, \pi)$, and $\beta \in (0, \pi)$, a
\emph{$\cset$-polygonal chain with orientation $\theta$}, $\csetc$, is
a chain with $2k - 1$ consecutive alternating links with slopes
$\tan(\theta)$ and $\tan(\theta + \beta)$ such that the extreme links
are half-lines with orientation $\tan(\theta)$. Let us define
$\ell_{i,\beta}(\theta)$ as the line passing through $p_i \in P$ with
slope $\tan(\theta + \beta)$. The \emph{fitting distance} between
$p_i$ and $\csetc$ is given by

$$
d_f(p_i,\csetc) = \underset{p \in \ell_{i,\beta}(\theta) \cap
  \csetc}{\text{min}} d(p_i,p),
$$
where $d(p_i,p)$ represents the Euclidean distance between $p_i$ and
$p$. The \emph{error tolerance} of $\csetc$ with respect to $P$ is the
maximum fitting distance between $\csetc$ and the elements in $P$,
that is

$$
\mu(\csetc,P) = \underset{p_i \in P} \max \hspace{0.3cm} d_f(p_i,\csetc).
$$
The $\cset$-fitting problem for $P$ with the Min-Max criterion,
consists on finding a polygonal chain $\csetc$ with minimum error
tolerance $\mu(\csetc,P)$. See Figure~\ref{apps:fitting:fig:fitting}.

\begin{theorem}[\cite{fitting_2011}]
  \label{apps:fitting:thm:2-fitting}
  The $\cset[2][\frac{\pi}{2}]$-fitting problem can be solved in
  $\Theta(n \log n)$ time and $O(n)$ space.
\end{theorem}

We consider here the case where $\theta$ has a constant value, namely
$0$, and we want to find the chain
$\csetc[2][][0] = \mathcal{C}_{2,\beta}$ of optimal error
tolerance.~More formally, we solve the following problem.

\begin{problem}[Oriented $\left( 2,\beta \right)$-fitting]
  Given a set $P$ of $n$ points in the plane, compute a polygonal
  chain $\mathcal{C}_{2,\beta}$ such that
  $\mu(\mathcal{C}_{2,\beta},P)$ has minimum value.
\end{problem}

Consider the algorithm used in~\cite{fitting_2011} to obtain the
$O(n \log n)$ time bound for the $\cset[2][\frac{\pi}{2}]$-fitting
problem used to prove \Cref{apps:fitting:thm:2-fitting}.~The
$\bhull[\frac{\pi}{2}]$ of $P$ is used as a tool to solve the problem
in $O(\log n)$ time for a fixed value of $\theta$ in a closed
orientation interval $[\theta_i,\theta_{i+1}]$.~An event sequence of a
linear number of orientation intervals is created to maintain
$\bhullp[P][\frac{\pi}{2}]$ as $\theta$ grows from $0$ to $2\pi$.

To solve Problem~3 we can follow exactly the same techniques.~We refer
the reader to reference~\cite{fitting_2011} just to see the evident
changes coming from the use of a different structure.~More concretely,
the structure $\bhullp[P][\frac{\pi}{2}]$ is replaced by $\bhullp$
which needs also a linear number of \emph{interval events}
$[\beta_i,\beta_{i+1}]$ to be maintained, and where the angular sweep
is performed over~$\beta$.~Thus, Lemmas 3 and 4
in~\cite{fitting_2011} can be now stated as follows:\\

(i) \emph{Given a value $\beta\in [\beta_i,\beta_{i+1}]$, an optimal
  solution of the $\cset[2]$-fitting problem for $\beta$ is defined by
  a line $\ell_{i,\beta}$ with slope
  $\tan(\beta)$ passing through a point $p_i$ of $P$ which gives the bipartition of $P$}.\\

(ii) \emph{The optimal solution of the $\cset[2]$-fitting problem for
  an interval event $[\beta_i,\beta_{i+1}]$ occurs either at an
  endpoint of the interval, i.e., at $\beta_i$ or $\beta_{i+1}$, or at
  a value $\beta_0 \in [\beta_i,\beta_{i+1}]$
  when the left and right error tolerance are equal}.\\

Using the properties (i),(ii) and following the maintenance of
$\bhullp$, the problem is solved as follows:
\begin{enumerate}
\item Compute $\bhullp$ and the optimal error tolerance for the first
  interval between events.
\item Traverse the event sequence, obtaining the optimal error
  tolerance at each interval between events.
\item Update the previous solution only when it is improved.
\end{enumerate}

Thus, the approach and arguments used in
Theorem~\ref{apps:fitting:thm:2-fitting} hold in the case of the
$\cset[2]$-fitting problem See
Figure~\ref{apps:fitting:fig:fitting}.~As a consequence, we get the
following theorem.

\begin{theorem}
  \label{apps:fitting:thm:fitting}
  The $\cset[2]$-fitting problem can be solved in $O(n\log n)$
  time and $O(n)$ space.
\end{theorem}



\begin{figure}[H]
  \centering
  \begin{minipage}{0.9\textwidth}
    \centering
    {\includegraphics[scale=0.9]{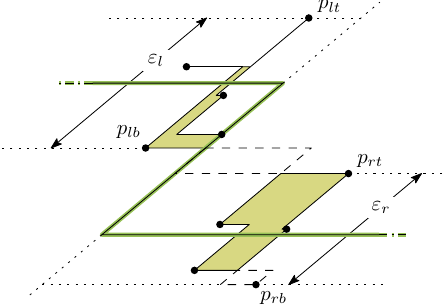}}
    \caption{The polygonal chain $\mathcal{C}_{2,\beta}$ and the
      $\bhull$ of $P$.}
    \label{apps:fitting:fig:fitting}
  \end{minipage}
\end{figure}

\end{document}